\documentclass{article} 
\usepackage{fullpage}
\usepackage{xargs}
\usepackage{nameref,hyperref,cleveref}
\hypersetup{
    colorlinks,
    citecolor=black,
    filecolor=black,
    linkcolor=blue,
    urlcolor=black
}
\usepackage{array}

\usepackage[normalem]{ulem}

\usepackage[encapsulated]{CJK} 
\usepackage[utf8]{inputenc} 
\usepackage{newunicodechar}

\usepackage{url}
\usepackage{eurosym}

\newcommand\synerror{\text{\bf * }}

\usepackage{amsmath}
\usepackage{amssymb}
\usepackage{amsthm}
\usepackage{stmaryrd}
\usepackage{pxfonts}
\usepackage{proof}
\usepackage{graphicx}
\usepackage[all]{xy}

\usepackage{thmtools}
\usepackage{thm-restate}

\newunicodechar{⊸}{\multimap}
\newunicodechar{⟜}{\text{\reflectbox{$\multimap$}}}
\newcommand\ImpL{⊸}
\newcommand\ImpR{⟜}
\newunicodechar{⇝}{\leadsto}
\newunicodechar{↭}{\leftrightsquigarrow}

\makeatletter
\def\@endtheorem{\endtrivlist}
\makeatother

\declaretheorem[name=Theorem,numberwithin=section]{theorem}
\declaretheorem[sibling=theorem]{lemma}
\declaretheorem[sibling=theorem]{corollary}

\declaretheorem[sibling=theorem]{proposition}

\declaretheorem[sibling=theorem]{definition}
\declaretheorem[sibling=theorem]{example}
\declaretheorem[sibling=theorem]{observation}

\newcommand\cod{\mathop{\mathrm{cod}}}

\newcommand\negL[2][\bot]{\mathop{\overset{#1}{\underset{#2}{⊸}}}}
\newcommand\negR[2][\bot]{\mathop{\overset{#1}{\underset{#2}{⟜}}}} 

\newcommand\op{ {\rm op} }
\newcommand\dom{\mathop{\rm dom}}

\renewcommand\oast{\mathop{\circledast}}

\newcommand\mul{\cdot}
\newcommand\emul{\cdot}

\newcommand\defeq{\stackrel{\text{\tiny def}}{=}}

\newcommand\id{-}

\def\c#1{ {\mathbf #1} }

\newcommand\refs{\sqsubset}

\newcommand\morph[1]{\overset{#1}\longrightarrow}

\newcommand\seq[1]{\underset{#1}\longrightarrow}
\newcommand\nseq{\longrightarrow}

\newcommand\resetL{(\rc{\id}\mul\id);\plugL}

\newcommand*{\longhookrightarrow}{\ensuremath{\lhook\joinrel\relbar\joinrel\rightarrow}}
\newcommand*{\longtwoheadrightarrow}{\ensuremath{\relbar\joinrel\twoheadrightarrow}}
\newcommand\pullseq[1]{\underset{#1}\longhookrightarrow}
\newcommand\pushseq[1]{\underset{#1}\longtwoheadrightarrow}

\newcommand\Set{ {\mathbf{Set} } }
\newcommand\SubSet{ {\mathbf{SubSet} } }

\newcommand\Cat{ {\mathbf{Cat} } }

\newcommand\SubCat{ {\mathbf{SubCat} } }

\newcommand\deq{\sim}

\let\vd=\vdash

\newcommand\Wand{\mathbin{\relbar\joinrel{\!*}}}
\newcommand\WandL{\mathbin{\relbar\joinrel{\!*}}}
\newcommand\WandR{\mathbin{*\!\joinrel\relbar}}
\newcommand\negWandL[2][\bot]{\mathop{\overset{#1}{\underset{#2}{\WandL}}}}
\newcommand\negWandR[2][\bot]{\mathop{\overset{#1}{\underset{#2}{\WandR}}}} 

\newcommand\plugL{\circledgtr}
\newcommand\plugR{\circledless}

\newcommand\pull[1]{{#1}^*\,}
\newcommand\push[1]{{#1}\,}

\newcommand\lc[1]{\lambda[{#1}]}
\newcommand\rc[1]{\rho[{#1}]}

\newcommand\G{\Gamma}

\begin{document}

\title{Type refinement and monoidal closed bifibrations}

\author{Paul-André Melliès \qquad Noam Zeilberger}



\maketitle

\begin{abstract}
The concept of {\em refinement} in type theory is a way of reconciling the ``intrinsic'' and the ``extrinsic'' meanings of types.  We begin with a rigorous analysis of this concept, settling on the simple conclusion that the type-theoretic notion of ``type refinement system'' may be identified with the category-theoretic notion of ``functor''.  We then use this correspondence to give an equivalent type-theoretic formulation of Grothendieck's definition of (bi)fibration, and extend this to a definition of {\em monoidal closed bifibrations}, which we see as a natural space in which to study the properties of proofs and programs.  Our main result is a representation theorem for strong monads on a monoidal closed fibration, describing sufficient conditions for a monad to be isomorphic to a continuations monad ``up to pullback''.
\end{abstract}

\tableofcontents

\pagebreak

\section{Introduction}

One of the difficulties in giving a clear mathematical definition of the ``topic'' of type theory is that the word ``type'' is actually used with two very different intuitive meanings and technical purposes in mind:

\begin{enumerate}
\item Like the syntactician's {\em parts of speech}, as a way of defining the grammar of well-formed expressions.
\item Like the semanticist's {\em predicates}, as a way of identifying subsets of expressions with certain desirable properties.
\end{enumerate}
These two different views of types are often associated respectively with Alonzo Church and Haskell Curry (hence ``types à la Church'' and ``types à la Curry''), while the late John Reynolds referred to these as the {\em intrinsic} and the {\em extrinsic} interpretations of types \cite{reynolds9x}.  In the intrinsic view, all expressions carry a type, and there is no need (or even sense) to consider the meaning of ``untyped'' expressions; while in the extrinsic view, every expression carries an independent meaning, and typing judgments serve to assert some property of that meaning.

Usually, readings of type theory through the lens of category theory have sided towards the intrinsic view.  This is natural given the analogy
\begin{center}
type system $\sim$ category
\end{center}
which says for example that a judgment
$$
x_1:A_1,\dots,x_n:A_n \vd e : B
$$
of the simply-typed lambda calculus may be interpreted as a morphism
$$
A_1 \times \dots \times A_n \morph{e} B
$$
in a cartesian-closed category \cite{lambekscott8x}.  This favors the intrinsic interpretation, since any morphism of a category
$$
A \morph{f} B
$$
is intrinsically associated with a pair of types (or ``objects''), namely, its domain $\dom(f) = A$ and codomain $\cod(f) = B$. Nor is it considered sensible to write the {\em same} morphism between a different pair of objects, 
\begin{align*}
A &\morph{f} B \\
\synerror{}  A' &\morph{f} B'
\end{align*}
although it's possible to have different morphisms
\begin{align*}
A &\morph{f} B \\
A &\morph{g} B
\end{align*}
between the same pair of objects.

But while the identification of typing judgments with morphisms of a category works nicely for systems like the simply-typed lambda calculus, for better or worse, the extrinsic view of types is also an important aspect of type theory, and does not seem to sit well with this analogy.  For instance, certain basic type-theoretic notions such as {\em intersection types} and {\em subtyping} really call out for an extrinsic reading.  Typically, the most natural and direct reading of the {\em intersection introduction rule}
$$
\infer{\G \vd e : A \cap B}{\G \vd e : A & \G \vd e : B}
$$
makes different typing judgments about the same expression, as does the most natural reading of the {\em subsumption rule}
$$
\infer{\G \vd e : B}{\G \vd e : A & A \le B}
$$
Trying to give an intrinsic interpretation of these rules instead (e.g., by asserting the existence of ``hidden coercions'') requires mental gymnastics.

Indeed, the very idea of a typing {\em judgment} in some sense presupposes a domain of expressions which may be {\em judged}.  Per Martin-Löf gave an influential dissection of the concept of judgment in logic \cite{martinlof8x}, and in many ways his theory of dependent types is all about the {\em interplay} of intrinsic and extrinsic interpretations. This is to say that the mismatch between categories and type systems cannot be dismissed as a mere artifact of notation---instead it reveals that something is {\em conceptually} missing in the standard categorical reading of type theory.

\section{Type refinement systems, fibrations and bifibrations}

We want to offer a different reading, and our starting point will be a categorical analysis of the concept of {\em refinement} in type theory, which is a way of reconciling the intrinsic and the extrinsic meanings of types \cite{freemanpfenning91,pfenning08}.  The basic idea of refinement is simple: a ``type à la Curry'' should not be considered as a predicate in a vacuum, but really as a predicate over a given ``type à la Church''.  In the limiting case, perhaps, there is a {\em unique} underlying ``intrinsic type'' which all of the different ``extrinsic types'' refine, but most often one's world is more diverse, and it is helpful to keep this in mind.

Our main aim in this section is to explain how the analogy
\begin{center}
type system $\sim$ category
\end{center}
may be generalized to an analogy
\begin{center}
type refinement system $\sim$ functor
\end{center}
and to then use this analogy to give an equivalent type-theoretic reformulation of Grothendieck's definition of {\em fibration} and {\em bifibration}.

\subsection{Reading a functor as a type refinement system}
\label{sec:reffun}

Let us suppose given two categories $\c{I}$ and $\c{E}$, related by a functor $p : \c{E} \to \c{I}$.  We establish a few terminological and notational conventions.

We refer to the objects of $\c{I}$ as {\bf i-types} $A,B,\dots$, and to its morphisms as {\bf expressions} $f,g,\dots$.  We indicate the signature of an expression in the traditional categorical style by writing the expression above an arrow from its domain to its codomain,
$$
A \morph{f} B
$$
or else using the type-theoretic colon notation $f : A \to B$.  Expressions are composed in diagrammatic order, i.e., we write the composition of
$$A \morph{f} B\quad\text{and}\quad B \morph{g} C$$
as
$$A \morph{f ; g} C$$
We indicate the identity morphism on an i-type $A$ as the expression $\id_A$, or often simply ``$\id$'' when the i-type is clear from context.

We refer to the objects of $\c{E}$ as {\bf e-types} $S,T,\dots$ and to its morphisms as {\bf derivations} $\alpha,\beta,\dots$.  Otherwise, we keep the same notational conventions for e-types and derivations as for i-types and expressions, writing 
the composition of
$$S \morph{\alpha} T\quad\text{and}\quad T \morph{\beta} U$$
as 
$$S \morph{\alpha ; \beta} U$$
and the identity derivation on an e-type $S$ by $\id_S$.

\begin{definition}\label{defn:refine}
We say that an e-type $S$ {\bf refines} an i-type $A$, written $S \refs A$, if $p(S) = A$.
\end{definition}
Now, suppose given an expression $f : A \to B$ and two e-types $S \refs A$ and $T \refs B$.  Such a triple of information is called a {\bf typing judgment}, which we notate by writing $f$ \uline{below} an arrow from $S$ to $T$:
$$S \seq{f} T$$
In the special case where $A = B$ and $f = \id_A$, we use the abbreviated notation
$$S \nseq T \defeq S \seq{\id_A} T$$
which we call a {\bf subtyping judgment}.
\begin{definition}
A {\bf typing derivation} for a (sub)typing judgment $S \seq{f} T$ is a derivation  $\alpha : S \to T$ such that $p(\alpha) = f$.  We notate this concisely by placing $\alpha$ over the judgment:
$$\deduce{S \seq{f} T}{\alpha}$$
A (sub)typing judgment is said to be {\bf derivable} if there exists a typing derivation for that judgment.  We notate this with a turnstile to the left of the judgment:
$$\vdash S \seq{f} T$$
\end{definition}
We will adapt the standard conventions of proof theory in using {\em inference rules} as a compact notation for generating typing derivations.  Somewhat informally, we say that an inference rule
$$
\infer{S \seq{f} T}{S_1 \seq{f_1} T_1 & \dots & S_n \seq{f_n} T_n}
$$
is {\em admissible} if there is an operation $D$ for transforming derivations of the premises
$$
\deduce{S_1 \seq{f_1} T_1}{\alpha_1} \quad \cdots \quad 
\deduce{S_n \seq{f_n} T_n}{\alpha_n}
$$
into a derivation of the conclusion
$$
\deduce{S \seq{f} T}{D(\alpha_1,\dots,\alpha_n)}
$$
We will often also label an admissible rule with the corresponding operation on derivations, as an annotation to the side of the horizontal line:
$$
\infer[D]{S \seq{f} T}{S_1 \seq{f_1} T_1 & \dots & S_n \seq{f_n} T_n}
$$
For example, {\em composition} and {\em identity} typing rules are admissible:
$$
\infer[C]{S \seq{f;g} U}{S \seq{f} T & T \seq{g} U}
\qquad
\infer[I]{S \seq{\id} S}{}
$$
In particular, the operation $C$ is defined by $C(\alpha,\beta) = (\alpha ; \beta)$, while $I$ is defined by $I = \id_S$; the fact that these rules are admissible is immediate from the assumption that $p : \c{E} \to \c{I}$ is a functor.  Likewise, {\em reflexivity}, {\em transitivity}, and {\em subsumption} rules for subtyping are admissible, 
$$
\infer{S \nseq S}{}
\qquad
\infer{S \nseq U}{S \nseq T & T \nseq U}
$$

$$
\infer{S \seq{f} U}{S \seq{f} T & T \nseq U}
\quad
\infer{S \seq{g} U}{S \nseq T & T \seq{g} U}
$$
noting that reflexivity is by definition just another way of writing the identity typing rule $I$, and that transitivity and subsumption are all special cases of $C$ with one or both of $f$ and $g$ set to $\id$.

Since $\c{I}$ is a category, there is a notion of identity of expressions, which we notate $f \deq g$.  We allow ourselves to treat typing judgments modulo identity of expressions, so that
$$
S \seq{f} T\quad\text{and}\quad S \seq{g} T
$$
are considered interchangeable for expressions $f \deq g$, although for clarity we may sometimes indicate the move between them as a {\em conversion rule},
$$
\infer[\deq]{S \seq{g} T}{S \seq{f} T}
$$
Since $\c{E}$ is also a category, there is likewise a notion of equality of derivations.  However, we won't typically refer to equality between ``naked'' derivations, but only between derivations of particular typing judgments.  For example, the associativity and unit equations of $\c{E}$ imply the following equations betwen derivations of typing judgments:
\begin{itemize}
\item (associativity)
$$
\infer[C]{S \seq{(f;g);h} V}{\infer[C]{S \seq{f;g} U}{\deduce{S \seq{f} T}{\alpha} & \deduce{T \seq{g} U}{\beta}} & \deduce{U \seq{h} V}{\gamma}}
\quad\deq\quad
\infer[C]{S \seq{f;(g;h)} V}{\deduce{S \seq{f} T}{\alpha} & \infer[C]{T \seq{g;h} V}{\deduce{T \seq{g} U}{\beta} & \deduce{U \seq{h} V}{\gamma}}}
$$
\item (right unit)
$$
\infer[C]{S \seq{(f;g);\id} U}{\infer[C]{S \seq{f;g} U}{\deduce{S \seq{f} T}{\alpha} & \deduce{T \seq{g} U}{\beta}} & \infer[I]{U \seq{\id} U}{}}
\quad\deq\quad
\infer[C]{S \seq{f;g} U}{\deduce{S \seq{f} T}{\alpha} & \deduce{T \seq{g} U}{\beta}}
$$
\item (left unit)
$$
\infer[C]{T \seq{\id;(g;h)} V}{\infer[I]{T \seq{\id} T}{} & \infer[C]{T \seq{g;h} V}{\deduce{T \seq{g} U}{\beta} & \deduce{U \seq{h} V}{\gamma}}}
\quad\deq\quad
\infer[C]{T \seq{g;h} V}{\deduce{T \seq{g} U}{\beta} & \deduce{U \seq{h} V}{\gamma}}
$$
\end{itemize}
Finally, besides the usual notion of isomorphism of objects of $\c{E}$, we can consider a stronger notion of ``vertical'' isomorphism of e-types.
\begin{definition}
We say that two e-types $S,T \refs A$ refining a common i-type are {\bf vertically isomorphic} ($S \deq T$) when there exist a pair of subtyping derivations
$$
\deduce{S \nseq T}\alpha
\qquad
\deduce{T \nseq S}\beta
$$
which compose to reflexivity
$$
\infer[C]{S \nseq S}{\deduce{S \nseq T}\alpha & \deduce{T \nseq S}\beta}
\quad\deq\quad
\infer[I]{S \nseq S}{}
\qquad
\infer[C]{T \nseq T}{\deduce{T \nseq S}\beta & \deduce{S \nseq T}\alpha}
\quad\deq\quad
\infer[I]{T \nseq T}{}
$$
\end{definition}
In the sequel, whenever we say that two e-types are isomorphic, we really mean vertical isomorphism. 

Now let us take a moment to reflect. Of course, everything we have said so far is completely trivial, a mere matter of changing some of the standard categorical terminology and establishing some syntactic conventions.  We hope the point that comes across, though, is that a lot of type-theoretic commentary can be extracted from the mere existence of a functor (if we only know where to look for the hidden soundtrack!).  In particular, these observations motivate our adopting the following simple definition.
\begin{definition}A {\bf type refinement system} is just a functor $p : \c{E} \to \c{I}$.
\end{definition}

\subsection{A typical example}
\label{sec:refexample}

To try to provide a bit of intuition for this funny way of reading functors, we will consider a simple and naive example, which is indeed perhaps the ``folk model'' for type refinement systems.  For $\c{I}$ we take the category of sets and functions $\Set$, while for $\c{E}$ we take the category of {\em subsets and image inclusions} $\SubSet$.  An object of $\SubSet$ is just a subset of a given underlying set
$$S \subseteq A$$
while a morphism
$$(S \subseteq A) \to (T \subseteq B)$$
is a function between the underlying sets
$$f : A \to B$$
such that the image of the first subset is included in the second
$$f(S) \subseteq T$$
As the functor $p : \SubSet \to \Set$, we take the forgetful map sending a subset $S \subseteq A$ to its underlying set $A$, and a function $f : A \to B$ to itself (simply forgetting the fact that $f(S) \subseteq T$).

By most interpretations, this model is already quite rich with i-types.  For example we might suppose it contains i-types of the natural numbers, integers, sequences of integers, 
$$
\mathbb{N}, \mathbb{Z}, \mathbb{Z}^{\mathbb{N}}
$$
and many more besides.  But the philosophy of type refinement is that rather than trying to translate every detail of the world into the language of $\c{I}$ (which is perhaps the traditional view of set-theoretic foundations), it is sometimes better to begin with a {\em rough} statement in $\c{I}$, then provide additional explanation in $\c{E}$.  Thus, for instance, we might consider the e-types of {\em odd} natural numbers or of {\em prime} natural numbers,
\begin{align*}
\{n \mid \exists k.n = 2k + 1\} &\refs \mathbb{N} \\
\{n \mid n\text{ prime}\} &\refs \mathbb{N}
\end{align*}
the e-types of {\em non-zero} integers or of  {\em non-negative} integers,
\begin{align*}
\{x \mid x \ne 0\} &\refs \mathbb{Z} \\
\{x \mid x \ge 0\} &\refs \mathbb{Z}
\end{align*}
the e-types of {\em linear} sequences or of {\em bounded} sequences,
\begin{align*}
\{f \mid \exists a,b \forall n. f(n) = a\cdot n + b\} &\refs \mathbb{Z}^{\mathbb{N}} \\
\{f \mid \exists x\forall n.f(n) \le x\} &\refs  \mathbb{Z}^{\mathbb{N}}
\end{align*}
and so on.  The point is that these e-types will always be considered with respect to the original i-types which they refine.

For example, the question whether ``every prime number is odd'' may be sensibly posed as a subtyping problem,
$$\{n \mid n\text{ prime}\} \nseq \{n \mid \exists k.n = 2k + 1\}$$
whose answer happens to be negative (i.e., the judgment is not derivable).  On the other hand, the question of whether ``every linear sequence is prime'' is not really sensible without resort to arbitrary conventions or encodings, and the corresponding subtyping judgment
$$\synerror \{f \mid \exists a,b \forall n. f(n) = a\cdot n + b\} \nseq \{n \mid n\text{ prime}\}$$
is not well-formed, since the two e-types refine different i-types.  As another example, if we take 
$$\lambda x.x^2 : \mathbb{Z} \to \mathbb{Z}$$
to be the squaring function on the integers, then the following three typing judgments are respectively derivable, underivable, and ill-formed:
\begin{align*}
\vdash \{x \mid x \ne 0\} &\seq{\lambda x.x^2} \{x \mid x \ge 0\}  \\
\nvdash \{x \mid x \ge 0\} &\seq{\lambda x.x^2}  \{x \mid x \ne 0\} \\
\synerror{}\{x \mid x \ne 0\} &\seq{\lambda x.x^2}  \{f \mid \exists x\forall n.f(n) \le x\}
\end{align*}

\subsection{Reading Grothendieck in translation}
\label{sec:fibref}

Let us recall the definition of when a functor $p : \c{E} \to \c{I}$ defines a {\em fibration} à la Grothendieck.
\begin{definition}
A morphism $\alpha : T' \to T$ in $\c{E}$ is said to be ($p-$){\bf cartesian} if for every object $S \in \c{E}$ and every pair of morphisms $\beta : S \to T$ and $g : p(S) \to p(T')$ such that $p(\beta) = g;p(\alpha)$, there is a unique morphism $\beta' : S \to T'$ such that $\beta = \alpha;\beta'$ and $p(\beta') = g$.  Let $f : A \to B$ be a morphism in $\c{I}$ and $T$ be an object of $\c{E}$ such that $p(T) = B$.  A morphism $\alpha$ in $\c{E}$ is said to be a {\bf cartesian lifting} of $f$ to $T$ if $p(\alpha) = f$, $\cod(\alpha) = T$, and $\alpha$ is cartesian.
\end{definition}
\begin{definition}
A functor $p : \c{E} \to \c{I}$ is said to be a {\bf fibration} if for every morphism $f : A \to B$ in $\c{I}$ and object $T \in \c{E}$ such that $p(T) = B$, $f$ has a cartesian lifting to $T$.
\end{definition}
\noindent
This definition may seem a bit mysterious to the uninitiated.  Rather than attempting to motivate it directly, we will now give an equivalent formulation in the language of  type refinement.  Again, we assume a fixed functor $p : \c{E} \to \c{I}$ and the notational and terminological conventions of \Cref{sec:reffun}.

\begin{definition}\label{def:pb}
Let $f : A \to B$ and $T \refs B$.  A ($p$-){\bf pullback of $T$ along $f$} is an e-type $\pull{f}T \refs A$ equipped with a pair of admissible rules
$$
\infer[L\pull{f}]{\pull{f}T \seq{f} T}{}
\qquad
\infer[R\pull{f}]{S \seq{g} \pull{f}T}{S \seq{g;f} T}
$$
referred to as the \uline{left rule} and the \uline{right rule}, such that for all derivations
$$\deduce{S \seq{g;f} T}{\beta}\quad\text{and}\quad\deduce{S \seq{g} \pull{f}T}{\eta}$$
we have equalities
$$
\infer[C]{S \seq{g;f} T}{
 \infer[R\pull{f}]{S \seq{g} \pull{f}T}{\deduce{S \seq{g;f} T}{\beta}} &
 \infer[L\pull{f}]{\pull{f}T \seq{f} T}{}
}
\quad\deq\quad
\deduce{S \seq{g;f} T}{\beta}
$$
and
$$
\deduce{S \seq{g} \pull{f}T}{\eta}
\quad\deq\quad
\infer[R\pull{f}]{S \seq{g} \pull{f}T}
{\infer[C]{S \seq{g;f} T}
{\deduce{S \seq{g} \pull{f}T}{\eta} & \infer[L\pull{f}]{\pull{f}T \seq{f} T}{}}}
$$
\end{definition}
\begin{proposition}\label{prop:pbuniq}
Any two pullbacks of $T$ along $f$ are isomorphic.
\end{proposition}
\begin{proof}
Let $T'$ and $T''$ both be pullbacks of $T$ along $f$, equipped with corresponding admissible rules
$$
\infer[LT']{T' \seq{f} T}{}
\qquad
\infer[RT']{S \seq{g} T'}{S \seq{g;f} T}
$$
and
$$
\infer[LT'']{T'' \seq{f} T}{}
\qquad
\infer[RT'']{S \seq{g} T''}{S \seq{g;f} T}
$$
We can build derivations of $T' \nseq T''$ and $T'' \nseq T'$ by
$$
\infer[RT'']{T' \nseq T''}{\infer[LT']{T' \seq{f} T}{}}
\qquad
\infer[RT']{T'' \nseq T'}{\infer[LT'']{T'' \seq{f} T}{}}
$$
Moreover, since
\small
\begin{center}
\begin{tabular}{m{4.8cm} m{10pt} m{4cm}}
$\infer[C]{T' \nseq T'}{
 \infer[RT'']{T' \nseq T''}{\infer[LT']{T' \seq{f} T}{}}
 &
 \infer[RT']{T'' \nseq T'}{\infer[LT'']{T'' \seq{f} T}{}}}
$
&$\deq$&
$\infer[RT']{T' \nseq T'}{
 \infer[C]{T' \seq{f} T}{
 \infer[C]{T' \nseq T'}{
 \infer[RT'']{T' \nseq T''}{\infer[LT']{T' \seq{f} T}{}}
 &
 \infer[RT']{T'' \nseq T'}{\infer[LT'']{T'' \seq{f} T}{}}
 }
 &
 \infer[LT']{T' \seq{f} T}{}
 }}$ \\\\
&$\deq$&
$\infer[RT']{T' \nseq T'}{
 \infer[C]{T' \seq{f} T}{
 \infer[RT'']{T' \nseq T''}{\infer[LT']{T' \seq{f} T}{}} &
 \infer[C]{T'' \seq{f} T'}{
 \infer[RT']{T'' \nseq T'}{\infer[LT'']{T'' \seq{f} T}{}}
 &
 \infer[LT']{T' \seq{f} T}{}
 }}}
$\\\\
&$\deq$&
$\infer[RT']{T' \nseq T'}{
 \infer[C]{T' \seq{f} T}{
 \infer[RT'']{T' \nseq T''}{\infer[LT']{T' \seq{f} T}{}} &
 \infer[LT'']{T'' \seq{f} T}{}
 }}$
\\\\
&$\deq$&
$\infer[RT']{T' \nseq T'}{
 \infer[LT']{T' \seq{f} T}{}
 }$\\\\
&$\deq$&
$\infer[I]{T' \nseq T'}{}$
\end{tabular}
\end{center}
\normalsize
and likewise (by a symmetric argument)
\begin{center}
\begin{tabular}{m{4.8cm} m{10pt} m{4cm}}
$
 \infer[C]{T'' \nseq T''}{
 \infer[RT']{T'' \nseq T'}{\infer[LT'']{T'' \seq{f} T}{}}
  &
  \infer[RT'']{T' \nseq T''}{\infer[LT']{T' \seq{f} T}{}}
 }$
&$\deq$&
$ \infer[I]{T'' \nseq T''}{}$
\end{tabular}
\end{center}
we have $T' \deq T''$ (and so we may speak of \emph{the} pullback $\pull{f}T$ when one exists).
\end{proof}
\begin{proposition}Whenever both sides exist, $\pull{(f;g)}T \deq \pull{f}\pull{g}T$
\end{proposition}
\begin{proof}
We construct derivations of $\pull{(f;g)}T \nseq \pull{f}\pull{g}T$ and $\pull{f}\pull{g}T\nseq \pull{(f;g)}T$ as
$$\infer[R\pull{f}]{\pull{(f;g)}T \nseq \pull{f}\pull{g}T}{
\infer[R\pull{g}]{\pull{(f;g)}T \seq{f} \pull{g}T}{
\infer[L\pull{(f;g)}]{\pull{(f;g)}T \seq{f;g} T}{
}}}
\qquad
\infer[R\pull{(f;g)}]{\pull{f}\pull{g}T \nseq \pull{(f;g)}T}{
\infer[C]{\pull{f}\pull{g}T \seq{f;g} T}{
\infer[L\pull{f}]{\pull{f}\pull{g}T \seq{f} \pull{g}T}{} &
\infer[L\pull{g}]{\pull{g}T \seq{g} T}{}
}}
$$
and again by an easy calculation, we can show that these two derivations compose to the identity.
\end{proof}
We write out these explicit proofs in order to demonstrate a certain style of argument (similar to reasoning in sequent calculus), but of course these properties of pullbacks are well-known.  Indeed, as the following proposition asserts, we have just dressed up Grothendieck's definition of cartesian liftings in type-theoretic notation.
\begin{proposition}\label{prop:pullcart} $\alpha : T' \to T$ is a cartesian lifting of $f$ to $T$ if and only if $T'$ is a pullback of $f$ along $T$, with the left rule given by $\alpha$, and the right rule defined by the universal property of $\alpha$.
\end{proposition}
\begin{proof}
Essentially immediate by unwinding the definitions.
\end{proof}
We can use this correspondence to restate the definition of when a functor is a fibration.
\begin{definition}
We say that a type refinement system {\bf has all pullbacks} if the pullback of $T$ along $f$ exists for every expression $f : A \to B$ and e-type $T \refs B$, or, to put it more concisely, if it is equipped with the following \uline{e-type formation rule}:
$$
\infer{\pull{f}T \refs A}{f : A \to B & T \refs B}
$$
\end{definition}
\begin{proposition}$p$ is a fibration iff $p$ has all pullbacks.
\end{proposition}
Now we give a similar makeover to Grothendieck's definition of {\em bifibration}.
\begin{definition}$p : \c{E} \to \c{I}$ is said to be a {\bf bifibration} if it is a fibration and if $p^\op : \c{E}^\op \to \c{I}^\op$ is also a fibration.
\end{definition}
\begin{definition}\label{def:pf}
Let $S \refs A$ and $f : A \to B$.  A ($p$-){\bf pushforward of $S$ along $f$} is an e-type $\push{f}S \refs B$ equipped with a pair of admissible rules
$$
\infer[Lf]{\push{f}S \seq{g} T}{S \seq{f;g} T}
\qquad
\infer[Rf]{S \seq{f} \push{f}S}{}
$$
such that for all derivations
$$\deduce{S \seq{f;g} T}{\beta}\quad\text{and}\quad\deduce{\push{f}S \seq{g} T}{\eta}$$
we have equalities
$$
\infer[C]{S \seq{f;g} T}{
 \infer[Rf]{S \seq{f} \push{f}S}{}
  &
 \infer[Lf]{\push{f}S \seq{g} T}{\deduce{S \seq{f;g} T}{\beta}}
}
\quad\deq\quad
\deduce{S \seq{f;g} T}{\beta}
$$
and
$$
\deduce{\push{f}S \seq{g} T}{\eta}
\quad\deq\quad
\infer[Lf]{\push{f}S \seq{g} T}
{\infer[C]{S \seq{f;g} T}
{\infer[Rf]{S \seq{f} \push{f}S}{} & \deduce{\push{f}S \seq{g} T}{\eta}}}
$$
\end{definition}
\begin{proposition}\label{prop:pbuniq}
Any two pushforwards of $S$ along $f$ are isomorphic.
\end{proposition}
\begin{proposition}Whenever both sides exist, $\push{(g;f)}T \deq \push{f}\push{g}T$
\end{proposition}
\begin{definition}
We say that a type refinement  system {\bf has all pushforwards} if it is equipped with the following e-type formation rule:
$$
\infer{\push{f}S \refs B}{S \refs A & f : A \to B}
$$
\end{definition}
\begin{proposition}$p$ is a bifibration iff $p$ has all pullbacks and pushforwards.
\end{proposition}
\begin{proposition}\label{prop:threeway}
In a bifibration we have a three-way correspondence of interderivability,
$$\vdash \push{f}S \nseq T \quad\text{iff}\quad \vdash S \seq{f} T \quad\text{iff}\quad \vdash S \nseq \pull{f}T$$
\end{proposition}
\begin{example}\label{ex:subsetbifib}
The type refinement system $\SubSet \to \Set$ of \Cref{sec:refexample} is a bifibration, where $\pull{f}T$ is the inverse image of $T$ under $f$, and $\push{f}S$ is the image of $S$ under $f$, i.e.,
\begin{align*}
\pull{f}T &\defeq \{ a \mid f(a) \in T \} \\
\push{f}S &\defeq \{f(a) \mid a \in S\}
\end{align*}
\end{example}
\begin{example}\label{ex:subcatbifib}
The previous example 
can be generalized by taking i-types to be categories rather than sets, and e-types to be \uline{presheaves} rather than subsets.  Putting aside issues of ``size'', $\Cat$ has categories $A,B,\dots$ as objects and functors $f : A \to B$ as morphisms, while $\SubCat$ has set-valued functors $S \in \Set^A$ as objects, and morphisms $(S : \Set^A) \to (T : \Set^B)$ given by pairs of a functor $f : A \to B$ together with a natural transformation $S \Rightarrow T\circ f$,
$$
\xymatrix @-1.5pc {
A \ar[rr]^S\ar[dd]_f && \Set \\
&\Rightarrow & \\
B \ar@/_1pc/[uurr]_{T}
}
$$
Then in diagrammatic terms, pullback and pushforward are defined respectively by precomposition and by \uline{left kan extension},
$$
\xymatrix @-1.5pc {
B \ar[rr]^T && \Set \\
& &\\
A\ar[uu]^f\ar[uurr]_{\pull{f}T}
}
\qquad
\xymatrix @-1.5pc {
A \ar[rr]^S\ar[dd]_f && \Set \\
&\Rightarrow & \\
B \ar@/_1pc/[uurr]_{\push{f}S}
}
$$
One can also describe the pullback and pushforward in pointwise form, 
\begin{align*}
\pull{f}T &\defeq a \mapsto T(fa) \\
\push{f}S &\defeq b \mapsto \int^a B(fa,b) \times T(a)
\end{align*}
where the formula for the pushforward denotes a \uline{coend}.
\end{example}
\begin{example}\label{ex:trivbifib}
Any category $\c{C}$ defines a bifibration over the trivial category $1$, with $\c{C} \to 1$ the functor which collapses all objects and arrows. Since there is only the identity arrow in $1$, vacuously all pullbacks and pushforwards exist.
\end{example}
\begin{example}\label{ex:hoarelogic}
A well-known example of a bifibration in computer science comes from taking $\c{I}$ to be a category of \uline{states and state transformers}, and $\c{E}$ a category of \uline{state predicates and valid assertions}.  In particular, a typing judgment may be read exactly like a ``Hoare triple'',
$$S \seq{f} T\ \ \sim\ \  \{P\}c\{Q\}$$
that is, as an assertion that the program $f$ will transform a state satisfying the precondition $S$ into a state satisfying the postcondition $T$.  Indeed, the typing rule $C$ is exactly the rule of sequential composition in Hoare logic,
$$
\infer{S \seq{f;g} U}{S \seq{f} T & T \seq{g} U} \ \ \sim\ \ 
\infer{\{P\} c_1;c_2\{R\}}{\{P\} c_1\{Q\}  & \{Q\} c_2\{R\}}
$$
while pullbacks correspond to the calculation of \uline{weakest preconditions} and pushforwards to \uline{strongest postconditions}.
\end{example}

\subsection{Weighted intersections and unions}
\label{sec:intun}

Although we will not explore this further here, we remark that pullbacks and pushforwards could also be seen as instances of a more general notion of ``weighted'' intersection and union types.

\begin{definition}
Let $(f_i : A \to B_i)_{i \in I}$ be a collection of expressions and $(T_i\refs B_i)_{i \in I}$ a collection of e-types.  The {\bf $(f_i)$-weighted intersection of the $(T_i)$} is an e-type
\def\ETYPE{\bigcap_{i\in I}\pull{f_i}T_i}
$\ETYPE \refs A$ equipped with a collection of admissible rules
$$
\infer[L^\cap{\pull{f_i}}]{\ETYPE \seq{f_i} T_i}{}
$$
as well as an admissible rule
$$
\infer[R^\cap{\pull{f_i}}]{S \seq{g} \ETYPE}{\forall i\in I. & S \seq{g;f_i} T_i}
$$
such that for all collections of derivations
$$\left(\deduce{S \seq{g;f_i} T_i}{\beta_i}\right)_{i\in I} \quad
\text{and}\quad\deduce{S \seq{g} \ETYPE}{\eta}$$
we have
$$
\infer[C]{S \seq{g;f_i} T_i}{
 \infer[R^\cap\pull{f_i}]{S \seq{g} \ETYPE}{\forall i\in I. & \deduce{S \seq{g;f_i} T_i}{\beta_i}} &
 \infer[L^\cap\pull{f_i}]{\ETYPE \seq{f_i} T_i}{}
}
\quad\deq\quad
\deduce{S \seq{g;f_i} T_i}{\beta_i}
$$
and
$$
\deduce{S \seq{g} \ETYPE}{\eta}
\quad\deq\quad
\infer[R^\cap\pull{f_i}]{S \seq{g} \ETYPE}
{\forall i\in I. & \infer[C]{S \seq{g;f_i} T_i}
{\deduce{S \seq{g} \ETYPE}{\eta} & \infer[L^\cap\pull{f_i}]{\ETYPE \seq{f_i} T_i}{}}}
$$
Dually, let $(f_i : A_i \to B)_{i \in I}$ be a collection of expressions and $(S_i\refs A_i)_{i \in I}$ a collection of e-types.  The {\bf $(f_i)$-weighted union of the $(S_i)$} is an e-type
\def\ETYPE{\bigcup_{i\in I}\push{f_i}S_i}
$\ETYPE \refs B$ equipped with a collection of admissible rules
$$
\infer[R_\cup{{f_i}}]{S_i \seq{f_i} \ETYPE}{}
$$
as well as an admissible rule
$$
\infer[L_\cup{{f_i}}]{\ETYPE \seq{g} T}{\forall i\in I. & S_i \seq{f_i;g} U}
$$
such that for all collections of derivations
$$\left(\deduce{S_i \seq{f_i;g} T}{\beta_i}\right)_{i\in I} \quad
\text{and}\quad\deduce{\ETYPE \seq{g} T}{\eta}$$
we have
$$
\infer[C]{S_i \seq{f_i;g} T}{
 \infer[R_\cup{f_i}]{S_i \seq{f_i} \ETYPE}{}
 &
 \infer[L_\cup{f_i}]{\ETYPE \seq{g} S}{\forall i\in I. & \deduce{S_i \seq{f_i;g} T}{\beta_i}}
}
\quad\deq\quad
\deduce{S_i \seq{f_i;g} T}{\beta_i}
$$
and
$$
\deduce{\ETYPE \seq{g} T}{\eta}
\quad\deq\quad
\infer[L_\cup{f_i}]{\ETYPE \seq{g} T}
{\forall i\in I. & \infer[C]{S_i \seq{f_i;g} T}
{\infer[R_\cup{f_i}]{S_i \seq{f_i} \ETYPE}{} & \deduce{\ETYPE \seq{g} T}{\eta}}}
$$
\end{definition}
\begin{definition}We say that a type refinement system is {\bf bicomplete} if all weighted intersections and unions exist, i.e., if it is equipped with the following e-type formation rules:
$$
\infer{\bigcap_{i\in I} \pull{f_i}T_i \refs A}{\forall i \in I. & f_i : A \to B_i & T_i \refs B_i}
\qquad
\infer{\bigcup_{i\in I} \push{f_i} S_i \refs B}{\forall i \in I. & S_i \refs A_i & f_i : A_i \to B}
$$
\end{definition}
Pullbacks and pushforwards of course correspond to the weighted intersection/union of a singleton, while the usual ``unweighted''  notion of intersection and union can be seen as weighting by the identity.  For example, with the definitions
\begin{align*}
T_1 \cap T_2 &\defeq \pull{(\id)}T_1 \cap \pull{(\id)}T_2 \\
S_1 \cup S_2 &\defeq (\id)S_1 \cup (\id)S_2
\end{align*}
the following type formation and typing rules are admissible in any bicomplete type refinement system:
$$
\infer{T_1 \cap T_2 \refs A}{T_1 \refs A & T_2 \refs A}\qquad
\infer{S_1 \cup S_2 \refs B}{S_1 \refs B & S_2 \refs B}
$$

$$
\infer{T_1\cap T_2 \nseq T_1}{}\quad
\infer{T_1\cap T_2 \nseq T_2}{}\qquad
\infer{S \seq{g} T_1 \cap T_2}{S \seq{g} T_1 & S \seq{g} T_2}
$$

$$
\infer{S_1 \cup S_2 \seq{g} T}{S_1 \seq{g} T & S_2 \seq{g} T}\qquad
\infer{S_1 \nseq S_1 \cup S_2}{}\quad
\infer{S_2 \nseq S_1 \cup S_2}{}
$$
Weighted intersections and unions in their full generality may be seen as an abstraction of the programming concepts of \emph{variant record} and \emph{tagged union}.

\section{Monoidal type refinement systems}
\label{sec:monref}

So far we have assumed nothing about the base category of i-types, other than that it is a category.  In this interlude we consider refinement of {\em monoidal} categories of i-types, with a corresponding monoidal structure on e-types.  (In the next section, we will consider refinement of monoidal {\em closed} categories.)


Recall that a {\bf monoidal category} is a category $\c{D}$ equipped with a bifunctor
$$-\mul- : \c{D} \times \c{D} \to \c{D}$$
and an object $1$, satisfying associativity and unity equations up to coherent natural isomorphism,
$$
(A\mul B)\mul C \deq A\mul (B\mul C) \qquad
A \mul 1 \deq A \deq 1 \mul A
$$
A {\bf strong monoidal functor} between two monoidal categories $(\c{E}, \mul_{\c{E}}, 1_{\c{E}})$ and $(\c{D}, \mul_{\c{D}}, 1_{\c{D}})$ is a functor $F : \c{E} \to \c{D}$ equipped with natural isomorphisms
\begin{align*}
F[A] \mul_{\c{D}} F[B] &\deq F [A \mul_{\c{E}} B] \\
1_{\c{D}} &\deq F [1_{\c{E}}]
\end{align*}
which again satisfy associativity and unity conditions.

\begin{definition}A {\bf monoidal type refinement system} is just a strong monoidal functor $p : \c{E} \to \c{I}$.
\end{definition}
As before, we can use type-theoretic language to elaborate on this compact definition.

We will omit subscripts when referring to the respective monoidal structures on $\c{E}$ and $\c{I}$, since there is never ambiguity in the way we use them.  The object part of the strong monoidal functor $p$ may be expressed as the following e-type formation rules,\footnote{Since the refinement relation was defined as an equality $p(S) = A$ (\Cref{defn:refine}), pedantically we should be speaking of \emph{strict} rather than strong monoidal functors.  However, ``morally'' (i.e., as a non-``evil'' notion) we really think of the refinement relation as being defined up to coherent isomorphism $p(S) \deq A$, which is why we don't feel a need to draw distinctions here between strong and strict monoidal functors.}
$$
\infer{S_1 \emul S_2 \refs A_1 \mul A_2}{S_1 \refs A_1 & S_2 \refs A_2}
\qquad
\infer{1  \refs 1}{}
$$
while the arrow part can be expressed as the following typing rules:
$$
\infer[M]{S_1 \emul S_2 \seq{f_1 \mul f_2} T_1 \emul T_2}
{S_1 \seq{f_1} T_1 & S_2 \seq{f_2} T_2}
\qquad
\infer[U]{1 \seq{1} 1}{}
$$
The equations of monoidal categories may be expressed as the following equations on derivations:
\begin{itemize}
\item (associativity)
$$
\infer[M]{(S_1 \emul S_2) \emul S_3 \seq{(f_1 \mul f_2) \mul f_3} (T_1 \emul T_2) \emul T_3}
{\infer[M]{S_1 \emul S_2 \seq{f_1 \mul f_2} T_1 \emul T_2}
  {\deduce{S_1 \seq{f_1} T_1}{\alpha_1} & \deduce{S_2 \seq{f_2} T_2}{\alpha_2}}
 & \deduce{S_3 \seq{f_3} T_3}{\alpha_3}}
\quad\deq\quad
\infer[M]{S_1 \emul (S_2 \emul S_3) \seq{f_1 \mul (f_2 \mul f_3)} T_1 \emul (T_2 \emul T_3)}
{\deduce{S_1 \seq{f_1} T_1}{\alpha_1} &
 \infer[M]{S_2 \emul S_3 \seq{f_2 \mul f_3} T_2 \emul T_3}{
   \deduce{S_2 \seq{f_2} T_2}{\alpha_2}
  & \deduce{S_3 \seq{f_3} T_3}{\alpha_3}}}
$$
\item (unit)
$$
\infer[M]{S \mul 1 \seq{f \mul 1} T \mul 1}
{\deduce{S \seq{f} T}{\alpha}  & \infer[U]{1 \seq{1} 1}{}}
\quad\deq\quad
\deduce{S \seq{f} T}{\alpha}
\quad\deq\quad
\infer[M]{1 \mul S \seq{1 \mul f} 1 \mul T}
{\infer[U]{1 \seq{1} 1}{} & \deduce{S \seq{f} T}{\alpha}}
$$
\item (bifunctoriality)
$$
\infer[C]{S_1 \emul S_2 \seq{(f_1\mul f_2);(g_1 \mul g_2)} U_1 \emul U_2}{
\infer[M]{S_1 \emul S_2\seq{f_1 \mul f_2} T_1 \emul T_2}{\deduce{S_1 \seq{f_1} T_1}{\alpha_1} & \deduce{S_2 \seq{f_2} T_2}{\alpha_2}} &
\infer[M]{T_1 \emul T_2\seq{g_1 \mul g_2} U_1 \emul U_2}{\deduce{T_1 \seq{g_1} U_1}{\beta_1} & \deduce{T_2 \seq{g_2} U_2}{\beta_2}}}
\quad\deq\quad
\infer[M]{S_1 \emul S_2 \seq{(f_1; g_1)\mul (f_2;g_2)} U_1 \emul U_2}{
\infer[C]{S_1 \seq{f_1;g_1} U_1}{\deduce{S_1 \seq{f_1} T_1}{\alpha_1} & \deduce{T_1 \seq{g_1} U_1}{\beta_1}} &
\infer[C]{S_1 \seq{f_2;g_2} U_2}{\deduce{S_2 \seq{f_2} T_2}{\alpha_2} & \deduce{T_2 \seq{g_2} U_2}{\beta_2}}}
$$
$$
\infer[I]{S \mul T \seq{\id} S \mul T}{}
\quad\deq\quad
\infer[M]{S \mul T \seq{\id\mul\id} S \mul T}{\infer[I]{S \seq{\id}S}{} & \infer[I]{T\seq{\id} T}{}}
\qquad
\infer[I]{1 \seq{\id} 1}{} \quad\deq\quad 
\infer[U]{1 \seq{1} 1}{}
$$
\end{itemize}

\begin{definition}
A {\bf monoidal (bi)fibration} is a monoidal type refinement system $p : \c{E} \to \c{I}$ with all pullbacks (and pushforwards), such that the monoidal product on $\c{E}$ preserves these pullbacks (and pushforwards).
\end{definition}
The fact that $p$ preserves all pullbacks and pushforwards may be expressed by saying that the canonical derivations
$$
\infer[R\pull{(f_1\mul f_2)}]{\pull{f_1}T_1 \cdot \pull{f_2}T_2 \nseq \pull{(f_1\mul f_2)} (T_1\mul T_2)}{
\infer[M]{\pull{f_1}T_1 \cdot \pull{f_2}T_2 \seq{f_1\mul f_2} T_1\mul T_2}{
\infer[L\pull{f_1}]{\pull{f_1}T_1 \seq{f_1} T_1}{} &
\infer[L\pull{f_2}]{\pull{f_2}T_2 \seq{f_2} T_2}{}
}}
\qquad
\infer[L\push{(f_1\mul f_2)}]{\push{(f_1\mul f_2)} (S_1\mul S_2) \nseq \push{f_1}S_1 \cdot \push{f_2}S_2}{
\infer[M]{S_1\mul S_2 \seq{f_1\mul f_2} \push{f_1}S_1 \cdot \push{f_2}S_2}{
\infer[R\push{f_1}]{S_1 \seq{f_1} \push{f_1}S_1}{} &
\infer[R\push{f_2}]{S_2 \seq{f_2} \push{f_2}S_2}{}
}}
$$
have inverses
$$
\infer{\pull{(f_1\mul f_2)} (T_1\mul T_2)\nseq \pull{f_1}T_1 \cdot \pull{f_2}T_2 }{} 
\qquad
\infer{\push{f_1}S_1 \cdot \push{f_2}S_2 \nseq \push{(f_1\mul f_2)} (S_1\mul S_2)}{}
$$
thereby witnessing the isomorphisms
\begin{align*}
\pull{f_1}T_1 \cdot \pull{f_2}T_2 &\deq \pull{(f_1\mul f_2)} (T_1\mul T_2) \\
\push{f_1}S_1 \cdot \push{f_2}S_2 &\deq \push{(f_1\mul f_2)} (S_1\mul S_2)
\end{align*}
We remark that this definition of monoidal fibration is essentially equivalent to the one appearing in  \cite{shulman08}. 

\section{Monoidal closed type refinement systems}
\label{sec:monclosedref}

In this section we work out the definition of \emph{monoidal closed bifibrations}---as a generalization of monoidal closed categories and a natural extension of the concept of bifibration---and describe some examples.  These will include examples of monoidal closed bifibrations, of course, but also examples of logical structures that can be naturally defined {\em inside} monoidal closed bifibrations.

Before we begin, though, it's worth spending a moment to discuss Lawvere's notion of \emph{hyperdoctrine} \cite{lawvere69}.  There is obviously a very close kinship between the approach we have been describing here and the principles behind hyperdoctrines.  What we call ``i-types'' correspond to what Lawvere just called ``types'', and what we call ``e-types'' correspond to what Lawvere called ``attributes'';\footnote{We prefer to emphasize that these are both aspects of the word \emph{type} as it has been employed in type theory.} pullback and pushforward correspond directly to ``substitution'' and ``existential quantification''.  However, besides the obvious difference that we choose to work in a monoidal rather than a cartesian setting, the crucial difference is in the way that the closed structure on i-types is used, and the closed structure on e-types introduced.

In contrast to the situation with hyperdoctrines, monoidal closed type refinement systems follow a sort of \emph{microcosm principle} \cite{baezdolan98}: in order to \emph{define} what it means for e-types to have a closed structure, the category of i-types already has to be monoidal closed.  And as we will see (\Cref{sec:examplesin,sec:repmonad}), the really interesting phenomena arise through the \emph{interaction} of the two monoidal closed structures---typically by forming a product or residual of e-types, and then pushing forward or pulling back along a map defined using the monoidal closed structure of $\c{I}$.

\subsection{Review of monoidal closed categories}
\label{sec:mccs}

Let $A,B$, and $C$ be objects of a monoidal category.  A {\bf left residual} of $C$ by $A$ is an object $\negL[C]{A}$ equipped with a map
$$\plugL : A \mul \negL[C]{A} \to C$$
and a natural transformation $\lambda$ from maps
$$A \mul B \to C$$
to maps
$$B \to \negL[C]{A}$$
such that for all $f : A \mul B \to C$ and $g : B \to \negL[C]{A}$ we have
\begin{align*}
(\id\mul \lc{f});\plugL &\deq f \\
g &\deq \lc{(\id\mul g);\plugL}
\end{align*}
Similarly, a {\bf right residual} of $C$ by $B$ is an object $\negR[C]{B}$ equipped with a map
$$\plugR : \negR[C]{B} \mul B \to C$$
and a natural transformation $\rho$ from maps
$$A \mul B \to C$$
to maps
$$A \to \negR[C]{B}$$
such that for all $f : A \mul B \to C$ and $g : A \to \negR[C]{B}$ we have
\begin{align*}
(\rc{f}\mul \id);\plugR &\deq f \\
g &\deq \rc{(g\mul \id);\plugR}
\end{align*}
A {\bf monoidal closed category} is a monoidal category equipped with left and right residuals for each pair of objects.  We remark that the following maps are definable in any monoidal closed category (we will use them in \Cref{sec:repmonad}) :
\begin{align*}
\tag{shift}  \lc{\plugR} &: B \to \negL[C]{\negR[C]{B}} \\
\tag{reset} \resetL &: \negL[C]{\negR[B]{B}}  \to C
\end{align*}
Note that the shift maps are the units of the {\bf continuation monads} arising from adjunctions of the form
$$
\xymatrix {
\c{I}\ar@/^1pc/[rr]^{\negR[C]{\id}} & \bot & \c{I}^\op\ar@/^1pc/[ll]^{\negL[C]{\id}}
}
$$
for each object $C$ of a monoidal closed category $\c{I}$.

\subsection{Residuals of e-types}

Let $p : \c{E} \to \c{I}$ be a monoidal type refinement system over a monoidal closed category $\c{I}$.
\begin{definition}
Let $S \refs A$ and $U \refs C$.  A ($p$-){\bf left residual} of $U$ by $S$ is an e-type $\negL[U]{S} \refs \negL[C]{A}$ equipped with a pair of admissible rules
$$
\infer[L^⊸]{S \emul \negL[U]{S} \seq{\plugL} U}{}
\qquad
\infer[R^⊸]{T \seq{\lc{f}} \negL[U]{S}}{S \emul T \seq{f} U}
$$
such that for all derivations
$$
\deduce{S \emul T \seq{f} U}{\beta}
\quad\text{and}\quad
\deduce{T \seq{g} \negL[U]{S}}{\eta}
$$
we have equalities
$$
\infer[C]{S\emul T \seq{(\id\mul \lc{f});\plugL} U}{
 \infer[M]{S \emul T \seq{\id\mul\lc{f}} S \emul \negL[U]{S}}{
    \infer[I]{S \seq{\id} S}{} &
    \infer[R^⊸]{T \seq{\lc{f}} \negL[U]{S}}{\deduce{S \emul T \seq{f} U}{\beta}}
  }
   &
 \infer[L^⊸]{S \emul \negL[U]{S} \seq{\plugL} U}{}
}
\quad\deq\quad
\deduce{S \emul T \seq{f} U}{\beta}
$$
and
$$
\deduce{T \seq{g} \negL[U]{S}}{\eta}
\quad\deq\quad
 \infer[R^⊸]{T \seq{\rc{(\id\mul g);\plugL}} \negL[U]{S}}{
 \infer[C]{S \emul T \seq{(\id\mul g);\plugL} U}{
  \infer[M]{S \emul T \seq{\id\mul g} S \emul \negL[U]{S}}{
   \infer[I]{S \nseq S}{} &
   \deduce{T \seq{g} \negL[U]{S}}{\eta}
   } &
  \infer[L^⊸]{S \emul \negL[U]{S} \seq{\plugL} U}{}
  }}
$$
\end{definition}
\begin{definition}
Let $T \refs B$ and $U \refs C$.  A ($p$-){\bf right residual} of $U$ by $T$ is an e-type $\negR[U]{T} \refs \negR[C]{B}$ equipped with a pair of admissible rules
$$
\infer[L^⟜]{\negR[U]{T} \emul T \seq{\plugR} U}{}
\qquad
\infer[R^⟜]{S \seq{\rc{f}} \negR[U]{T}}{S \emul T \seq{f} U}
$$
such that for all derivations
$$
\deduce{S \emul T \seq{f} U}{\beta}
\quad\text{and}\quad
\deduce{S \seq{g} \negR[U]{T}}{\eta}
$$
we have equalities
$$
\infer[C]{S\emul T \seq{(\rc{f}\mul \id);\plugR} U}{
 \infer[M]{S \emul T \seq{\rc{f}\mul\id} \negR[U]{T} \emul T}{
    \infer[R^⟜]{S \seq{\rc{f}} \negR[U]{T}}{\deduce{S \emul T \seq{f} U}{\beta}}
    &
     \infer[I]{T \seq{\id} T}{}
  }
   &
 \infer[L^⟜]{\negR[U]{T} \emul T \seq{\plugR} U}{}
}
\quad\deq\quad
\deduce{S \emul T \seq{f} U}{\beta}
$$
and
$$
\deduce{S \seq{g} \negR[U]{T}}{\eta}
\quad\deq\quad
 \infer[R^⟜]{S \seq{\rc{(g\mul\id);\plugR}} \negR[U]{T}}{
 \infer[C]{S \emul T \seq{(g\mul\id);\plugR} U}{
  \infer[M]{S \emul T \seq{g\mul\id} \negR[U]{T} \emul T}{
   \deduce{S \seq{g} \negR[U]{T}}{\eta}
    &
    \infer[I]{T \nseq T}{}
   } &
  \infer[L^⟜]{\negR[U]{T} \emul T \seq{\plugR} U}{}
  }}
$$
\end{definition}

\begin{definition}A monoidal type refinement system over a monoidal closed category of i-types is said to be {\bf closed} if it is equipped with left and right residuals for all pairs of e-types, i.e., such that the following e-type formation rules are admissible:\footnote{
NB: the formation rules for residuals sometimes appear strange at first to people familiar with the ``rule of contravariant subtyping'' for function types (and who thus expect something like $A \refs S$ in the premise).  This seems to be due to the long tradition of conflating the concepts of refinement and subtyping.  For example, it is easy to show that for any collection of e-types 
$$
S_1,S_2 \refs A \quad
T_1,T_2 \refs B \quad
U_1,U_2 \refs C
$$
the following {\em subtyping rules} are admissible in a monoidal closed type refinement system:
$$
\infer{\negL[U_1]{S_1} \nseq \negL[U_2]{S_2}}{S_2 \nseq S_1 & U_1 \nseq U_2} \qquad
\infer{\negR[U_1]{T_1} \nseq \negR[U_2]{T_2}}{U_1 \nseq U_2 & T_2 \nseq T_1}
$$
}
$$
\infer{\negL[U]{S} \refs \negL[C]{A}}{S \refs A & U \refs C}
\qquad
\infer{\negR[U]{T} \refs \negR[C]{B}}{U \refs C & T \refs B}
$$
\end{definition}
\begin{definition}A {\bf monoidal closed (bi)fibration} is a monoidal closed type refinement system which is also a monoidal (bi)fibration.
\end{definition}

\subsection{Examples of monoidal closed bifibrations}
\label{sec:examplesof}

\begin{example}The bifibration $\SubSet \to \Set$ (\Cref{ex:subsetbifib}) is in fact monoidal closed, with the monoidal closed structure on $\Set$ corresponding to the usual cartesian closed structure,
\begin{align*}
A \mul B &\defeq A \times B\\
\negL[C]{A} &\defeq C^A  \\
\negR[C]{B} &\defeq C^B
\end{align*}
the monoidal structure on $\SubSet$ corresponding to cartesian product of subsets,
\begin{align*}
S\mul T &\subseteq A \times B \\
S\mul T &\defeq \{ (a,b) \mid a \in S, b \in T \}
\end{align*}
and the residuals of e-types defined by 
\begin{align*}
\negL[U]{S} &\subseteq C^A \\
\negL[U]{S} &\defeq \{ f \mid f(S) \subseteq U\} \\
\negR[U]{T} &\subseteq C^B \\
\negR[U]{T} &\defeq \{ f \mid f(T) \subseteq U\}
\end{align*}
Note that the left and right residuals (of both i-types and e-types) coincide in this example, since the monoidal products on $\c{I}$ and $\c{E}$ are cartesian monoidal.
\end{example}

\begin{example}
The bifibration $\SubCat \to \Cat$  (\Cref{ex:subcatbifib}) is in fact monoidal closed, with the monoidal closed structure on $\Cat$ again corresponding to the usual cartesian closed structure, the monoidal structure on $\SubCat$ corresponding to ``external product'' of presheaves,
\begin{align*}
S\mul T &: A \times B \to \Set \\
S\mul T &\defeq (a,b) \mapsto S(a) \times T(b)
\end{align*}
and the residuals defined as sets of natural transformations
\begin{align*}
\negL[U]{S} & : C^A \to \Set \\
\negL[U]{S} &\defeq [A,C](S, U \circ f) \\
\negR[U]{T} & : C^B \to \Set \\
\negR[U]{T} &\defeq [B,C](T, U \circ f)
\end{align*}
or equivalently as \uline{ends}:
\begin{align*}
\negL[U]{S} &\defeq f \mapsto \int_a S(a) \to U(f a) \\
\negR[U]{T} &\defeq f \mapsto \int_b T(b) \to U(f b)
\end{align*}
\end{example}

\begin{example}The trivial bifibration $\c{C} \to 1$ is of course also a trivial monoidal closed bifibration whenever $\c{C}$ is a monoidal closed category.
\end{example}

\subsection{Examples {\em in} monoidal closed bifibrations}
\label{sec:examplesin}

The class of ``Hoare logic bifibrations'' of \Cref{ex:hoarelogic} are not typically considered as monoidal closed bifibrations.  On the other hand, Reynolds and O'Hearn's \emph{separation logic} \cite{reynolds02} provides a nice example of a logical structure which can be naturally described \emph{internally} to a monoidal closed bifibration.  Suppose the category of i-types includes a monoid $H$ of ``heaps'':
\begin{align*}
H &: \c{I} \\
\oast &: H \mul H \to H \\
emp &: 1 \to H
\end{align*}
Heap assertions are modelled as different refinements of $H$.  In particular, the ``separating conjunction'' and ``magic wand'' connectives on heap assertions may be defined as follows:
\begin{align*}
S * T &\defeq \push{\oast} (S \mul T) \\
S \Wand T &\defeq \pull{\rc{\oast}} \negR[T]{S}
\end{align*}
Interpreting this signature in the monoidal closed bifibration $\SubSet \to \Set$ gives the usual set-theoretic semantics of separation logic:
\begin{align*}
S * T &= \{h_1 \oast h_2 \mid h_1 \in S, h_2 \in T\} \\
S \Wand T & = \{ h \mid  \forall h'. h' \in S \to h\oast h' \in T\}
\end{align*}
On the other hand, we can see that the internal definition is much more general.  For example, interpreting the signature in $\SubCat \to \Cat$ recovers the \emph{Day construction} for lifting a monoidal structure on a category to a monoidal closed structure on its category of presheaves:
\begin{align*}
S * T &= h \mapsto \int^{h_1,h_2} H(h_1\oast h_2,h) \times S(h_1) \times T(h_2) \\
S \Wand T & = h \mapsto \int_{h'} S(h') \to T(h\oast h')
\end{align*}
The next proposition describes the situation more abstractly.
\begin{proposition}\label{prop:starwand}
With the above definitions of the connectives $*$ and $\Wand$, any monoidal closed bifibration admits the following subtyping rules (where all of the variables range over refinements of $H$),
$$
\infer[M{*}]{S_1 * S_2 \nseq T_1 * T_2}{S_1 \nseq T_1 & S_2 \nseq T_2}
\qquad
\infer[R^{\Wand}]{S \nseq T \Wand U}{S * T \nseq U}
\qquad
\infer[L^{\Wand}]{(T \Wand U) * T \nseq U}{}
$$
satisfying the equations
$$
\infer[C]{S * T \nseq U}{
 \infer[M{*}]{S * T \nseq (T \Wand U) * T}{
    \infer[R^{\Wand}]{S \nseq T \Wand U}{\deduce{S * T \nseq U}{\beta}}
    &
     \infer[I]{T \nseq T}{}
  }
   &
 \infer[L^{\Wand}]{(T\Wand U) * T \nseq U}{}
}
\quad\deq\quad
\deduce{S * T \nseq U}{\beta}
$$
and
$$
\deduce{S \nseq T\Wand U}{\eta}
\quad\deq\quad
 \infer[R^{\Wand}]{S \nseq T \Wand U}{
 \infer[C]{S * T \nseq U}{
  \infer[M{*}]{S * T \nseq T \Wand U * T}{
   \deduce{S \nseq T \Wand U}{\eta}
    &
    \infer[I]{T \nseq T}{}
   } &
  \infer[L^{\Wand}]{(T \Wand U) * T \nseq U}{}
  }}
$$
\end{proposition}
\begin{proof}
We show how to build the rules:
$$
\infer[L{\oast}]{S_1 * S_2 \nseq T_1 * T_2}{
\infer[C]{S_1 \mul S_2 \seq{\oast} T_1 * T_2}{
 \infer[M]{S_1 \mul S_2 \nseq T_1 \mul T_2}{S_1 \nseq T_1 & S_2 \nseq T_2} &
 \infer[R{\oast}]{T_1 \mul T_2 \seq{\oast} T_1 * T_2}{}
}}
$$
$$
\infer[R\pull{\rc{\oast}}]{S \nseq T \Wand U}{
\infer[R^{\ImpR}]{S \seq{\rc{\oast}} \negR[U]{T}}{
\infer[C]{S \mul T \seq{\oast} U}{
\infer[R{\oast}]{S \mul T \seq{\oast} S * T}{} &
S * T \nseq U}}}
\qquad
\infer[L{\oast}]{(T \Wand U) * T \nseq U}{
\infer[\deq]{(T \Wand U) \mul T \seq{\oast} U}{
\infer[C]{(T \Wand U) \mul T \seq{(\rc{\oast}\mul\id);\plugR} U}{
  \infer[M]{(T \Wand U) \mul T \seq{\rc{\oast}\mul\id} \negR[U]{T} \mul T}{
   \infer[L\pull{\rc{\oast}}]{T \Wand U \seq{\rc{\oast}} \negR[U]{T}}{} & 
   \infer[I]{T \nseq T}{}} &
  \infer[L^{\ImpR}]{\negR[U]{T} \mul T \seq{\plugR} U}{}
 }
}}
$$
The equations then follow from the equations of monoidal closed bifibrations, by a long but straightforward calculation.
\end{proof}
\begin{corollary}For all $T \refs H$, the operations $\id*T$ and $T\Wand \id$ are adjoint in the sense that the $R^{\Wand}$ rule is invertible,
$$
\infer={S \nseq T \Wand U}{S * T \nseq U}
$$
\end{corollary}
In fact, this adjunction is independent of whether the i-type $H$ is an actual monoid (i.e., of whether the operations $\oast$ and $emp$ satisfy associativity and unit equations), and indeed it even extends to binary operations of arbitrary type.
\begin{proposition}
Given an operation $\oast : A \mul B \to C$, in any monoidal closed bifibration we have formation rules
$$
\infer{S * T \refs C}{S \refs A & T \refs B}
\qquad
\infer{\negWandL[U]{S} \refs B}{S \refs A & U \refs C}
\qquad
\infer{\negWandR[U]{T} \refs A}{U \refs C & T \refs B}
$$
where
\begin{align*}
S * T &\defeq \push{\oast} (S \mul T) \\
\negWandL[U]{S} &\defeq \pull{\lc{\oast}} \negL[U]{S} \\
\negWandR[U]{T} &\defeq \pull{\rc{\oast}} \negL[U]{T} 
\end{align*}
satisfying a three-way adjunction, 
$$
\infer={S \nseq \negWandR[U]{T}}{
\infer={S * T \nseq U}{
T \nseq \negWandL[U]{S}}}
$$
\end{proposition}

\section{Representing monads}
\label{sec:repmonad}

One of the original motivations for this study was to gain a better understanding of Andrzej Filinski's work on the representation of monadic effects in programming languages using continuations and state \cite{filinski94,filinski99}, and to place it in the wider context of universal algebra.

To a first approximation, Filinski's representation of monads using continuations is very similar in spirit to the so-called \emph{codensity monad} \cite{kock66,leinster12}.  The codensity monad of a functor $R : \c{D} \to \c{E}$ may be defined by the following end formula:
$$M^R[T] = \int_U \c{E}(T,R[U]) \to R[U]$$
In the case that $R$ has a left adjoint, then the codensity monad coincides with the monad induced by the adjunction, as a simple Yoneda-like calculation shows:
\begin{proposition}
If $L \dashv R$, then $M^R[T] \cong RL[T]$.
\end{proposition}
\begin{proof}
$$M^R[T] = \int_U  \c{E}(T,R[U]) \to R[U] \cong \int_U \c{D}(L[T],U) \to R[U] \cong RL[T]$$
\end{proof}
Logically speaking, the end formula describes the codensity monad as a sort of ``polymorphic double-negation'', in which the answer type (or ``falsehood'') is parameterized over the objects of a category.  In this sense, the coincidence $M^R[T] \cong RL[T]$ may be seen as merely a vast generalization of the tautology
$$
\vdash X \equiv \forall p.(X \supset p) \supset p
$$
of second-order logic, where the crucial step of proving the implication from right to left involves instantiating $p := X$ and applying the hypothesis $(X \supset X) \supset X$ to the trivial proof of $X \supset X$.

Our key insight was that ordinary double-negation becomes a sort of polymorphic double-negation after {\em pulling back along the double-negation introduction (shift) map}.  More precisely, we have the following fact:
\begin{observation}\label{obs:poly} For all e-types $T$ and $U$, if $V$ is a pullback of $U$ (along any $f$), then the subtyping judgment
$$
\pull{\lc{\plugR}}\negL[U]{\negR[U]{T}} \nseq \pull{\lc{\plugR}} \negL[V]{\negR[V]{T}}
$$
is derivable in a monoidal closed fibration.
\end{observation}
\begin{proof}
$$
\infer[R\pull{\lc{\plugR}}]{\pull{\lc{\plugR}}\negL[U]{\negR[U]{T}} \nseq \pull{\lc{\plugR}} \negL[{\pull{f}U}]{\negR[{\pull{f}U}]{T}}}{
\infer[R^\ImpL]{\pull{\lc{\plugR}}\negL[U]{\negR[U]{T}} \seq{\lc{\plugR}} \negL[{\pull{f}U}]{\negR[{\pull{f}U}]{T}}}{
\infer[R\pull{f}]{\negR[{\pull{f}U}]{T} \mul \pull{\lc{\plugR}}\negL[U]{\negR[U]{T}} \seq{\plugR} \pull{f}U}{
\infer[\deq]{\negR[{\pull{f}U}]{T} \mul \pull{\lc{\plugR}}\negL[U]{\negR[U]{T}} \seq{\plugR;f} U}{
\infer[C]{\negR[{\pull{f}U}]{T} \mul \pull{\lc{\plugR}}\negL[U]{\negR[U]{T}} \seq{\rc{\plugR;f}\mul \lc{\plugR};\plugL} U}{
 \infer[M]{\negR[{\pull{f}U}]{T} \mul \pull{\lc{\plugR}}\negL[U]{\negR[U]{T}} \seq{\rc{\plugR;f}\mul \lc{\plugR}} \negR[U]{T} \mul \negL[U]{\negR[U]{T}}}{
  \infer[R^\ImpR]{\negR[{\pull{f}U}]{T} \seq{\rc{\plugR;f}} \negR[U]{T}}{
   \infer[C]{\negR[{\pull{f}U}]{T}\mul T \seq{\plugR;f} U}{
    \infer[L^\ImpR]{\negR[{\pull{f}U}]{T}\mul T \seq{\plugR} \pull{f}U}{} &
     \infer[L\pull{f}]{\pull{f}U \seq{f} U}{}
    }}
  &
  \infer[L\pull{\lc{\plugR}}]{\pull{\lc{\plugR}}\negL[U]{\negR[U]{T}} \seq{\lc{\plugR}} \negL[U]{\negR[U]{T}}}{}
} & 
 \infer[L^\ImpL]{\negR[U]{T} \mul \negL[U]{\negR[U]{T}} \seq{\plugL} T}{}
}}}}}
$$
Note the crucial use (at the inference marked ``$\deq$'') of the identity
\begin{align*}
\rc{\plugR;f}\mul \lc{\plugR};\plugL &\ \deq\,\  \plugR;f
\end{align*}
which is valid in any monoidal closed category.
\end{proof}
Thus, double-negation into a particular type subsumes double-negation into all pullbacks of that type---provided we are in the context of a shift.

After a few preliminaries, we will show how this idea leads to a general representation theorem for strong monads on monoidal closed fibrations.

\subsection{Adjunctions and strong monads on type refinement systems}

From now on we will consider pairs of type refinement systems
$$
\xymatrix {
\c{E}\ar[d]_p & \c{D}\ar[d]^q \\
\c{I} & \c{J}
}
$$
To avoid heavy notation, we will keep the same conventions for $p$ and $q$ as we had when there was just a single type refinement system (writing, for example, $T \refs B$ for the refinement relation in $q$, rather than, say, $T \refs^q B$).  For clarity, though, we will distinguish the objects of $\c{D}$ as ``d-types'', and the objects of $\c{J}$ as ``j-types'', while continuing to refer to the objects of $\c{E}$ and $\c{I}$ as e-types and i-types.

\begin{definition}
Let $p : \c{E} \to \c{I}$ and $q : \c{D} \to \c{J}$ be a pair of type refinement systems.  
A {\bf morphism of type refinement systems} $L : p \to q$ is a pair of functors $L_0 : \c{I} \to \c{J}$ and $L_1 : \c{E} \to \c{D}$ forming a commuting square,
$$
\xymatrix {
\c{E}\ar[d]_p\ar[r]^{L_1} & \c{D}\ar[d]^q \\
\c{I}\ar[r]^{L_0} & \c{J}
}
$$
in the sense that the following rules are admissible (omitting subscripts):
$$
\infer{L[S] \refs L[A]}{S \refs A}
\qquad
\infer[L]{L[S] \seq{L[f]} L[T]}{S \seq{f} T}
$$
\end{definition}

\begin{definition}
An {\bf adjunction of type refinement systems} $p \dashv q$ is a pair of morphisms $(L_0,L_1) : p \to q$ and $(R_0,R_1) : q \to p$ together with a pair of adjunctions $(\eta_0,\epsilon_0) : L_0 \dashv R_0$ and $(\eta_1,\epsilon_1) : L_1 \dashv R_1$,
$$
\xymatrix {
\c{E}\ar@/^1pc/[rr]^{L_1}\ar[dd]_p & \bot & \c{D}\ar@/^1pc/[ll]^{R_1}\ar[dd]^q \\
\\
\c{I}\ar@/^1pc/[rr]^{L_0} & \bot & \c{J}\ar@/^1pc/[ll]^{R_0}
}
$$
which are compatible in the sense that the following rules are admissible,
$$
\infer[\eta]{S \seq{\eta} RL[S]}{}\qquad
\infer[\epsilon]{LR[T] \seq{\epsilon} T}{}
$$
and the following equations hold:
\begin{itemize}
\item (naturality)
$$
\infer[C]{S \seq{f;\eta} RL[T]}{\deduce{S \seq{f} T}{\alpha} & \infer[\eta]{T \seq{\eta} RL[T]}{}}
\quad\deq\quad
\infer[C]{S \seq{\eta;RL[f]} RL[T]}{\infer[\eta]{S \seq{\eta} RL[S]}{} &
 \infer[RL]{RL[S] \seq{RL[f]} RL[T]}{\deduce{S \seq{f} T}{\alpha}}}
$$
$$
\infer[C]{LR[T] \seq{\epsilon;f} U}{\infer[\epsilon]{LR[T] \seq{\epsilon} T}{} & \deduce{T \seq{f} U}{\beta}}
\quad\deq\quad
\infer[C]{LR[T] \seq{LR[f];\epsilon} U}{
   \infer[LR]{LR[T] \seq{LR[f]} LR[U]}{\deduce{T \seq{f} U}{\beta}}
   &
   \infer[\epsilon]{LT[U] \seq{\epsilon} U}{}}
$$
\item (triangle laws)
$$
\infer[C]{L[S] \seq{L[\eta];\epsilon} L[S]}{
 \infer[L]{L[S] \seq{L[\eta]} LRL[S]}{\infer[\eta]{S \seq{\eta} RL[S]}{}} &
 \infer[\epsilon]{LRL[S] \seq{\epsilon} L[S]}{}
}
\quad\deq\quad
\infer[I]{L[S] \nseq L[S]}{}
$$
$$
\infer[C]{R[T] \seq{\eta;R[\epsilon]} R[T]}{
 \infer[\eta]{R[T] \seq{\eta} RLR[T]}{} &
 \infer[R]{RLR[T] \seq{R[\epsilon]} R[T]}{\infer[\epsilon]{LR[T] \seq{\epsilon} T}{}}
}
\quad\deq\quad
\infer[I]{R[T] \nseq R[T]}{}
$$
\end{itemize}
\end{definition}
We will be interested in adjunctions that give rise to \emph{strong monads} on monoidal type refinement systems.  By this we mean that the induced monad $R_0L_0$ is strong in the usual sense of having a strength
$$\sigma : A \mul RL[B] \to RL[A \mul B]$$
compatible with the unit and multiplication, and that $R_1L_1$ is strong in a compatible way
$$
\infer[\sigma]{S \mul RL[T] \seq{\sigma} RL[S\mul T]}{}
$$
\begin{proposition}
In a monoidal closed type refinement system, every 
$U \refs C$ gives rise to an adjunction
$$
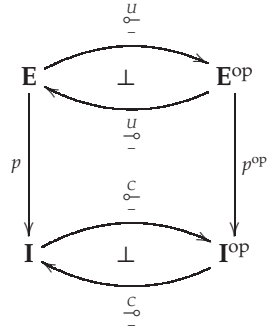
\xymatrix {
\c{E}\ar@/^1pc/[rr]^{\negR[U]{\id}}\ar[dd]_{p} & \bot & \c{E}^\op\ar@/^1pc/[ll]^{\negL[U]{\id}}\ar[dd]^{p^\op} \\
\\
\c{I}\ar@/^1pc/[rr]^{\negR[C]{\id}} & \bot & \c{I}^\op\ar@/^1pc/[ll]^{\negL[C]{\id}}
}
$$
and a corresponding strong monad on $p$.
\end{proposition}
Finally, we observe that by pulling back the monad $RL$ along the unit $\eta$, one obtains a ``fiberwise'' monad, meaning an operation
$$
\infer{\pull{\eta}RL[S] \refs A}{S \refs A}
$$
on the category $\c{E}_A$ of refinements of each i-type $A$, together with a pair of subtyping derivations
$$
\vdash S \nseq \pull{\eta}RL[S]\qquad
\vdash \pull{\eta}RL[\pull{\eta}RL[S]] \nseq \pull{\eta}RL[S]
$$
satisfying the monad laws.

\subsection{Diagrams of pullback and pushforward judgments}

Let $S \seq{f} T$ be a typing judgment.  By slight overloading of terminology, we say that the judgment itself is a pullback if $S$ is a pullback of $T$ along $f$, and indicate this by writing
$$
S \pullseq{f} T
$$
Similarly, we say that the judgment is a pushforward if $T$ is a pushforward of $S$ along $f$, indicated
$$
S \pushseq{f} T
$$
For example, a diagram
$$
S \pullseq{f} T \pullseq{g}U
$$
asserts that $S \deq \pull{f}T$ and $T \deq \pull{g}U$, while a diagram
$$
S \pushseq{f} T \pullseq{g}U
$$
asserts that $\push{f}S \deq T \deq \pull{g}U$.
\begin{proposition}\label{prop:2outof3}For all typing judgments $S \seq{f} T$ and $T \seq{g} U$ we have:
\begin{enumerate}
\item If $S \pullseq{f;g} U$  and $T \pullseq{g} U$ then $S \pullseq{f} T$.
\item If  $S \pushseq{f;g} U$ and $S \pushseq{f} T$ then $T \pushseq{g} U$.
\end{enumerate}
\end{proposition}
\begin{proof}
Since $S \deq \pull{(f;g)}U \deq \pull{f}\pull{g}U \deq \pull{f}T$ and
$U \deq \push{(f;g)} S \deq \push{g}\push{f}S \deq \push{g}T$.
\end{proof}

\subsection{The continuations representation of a monad}

We assume an adjunction
$$
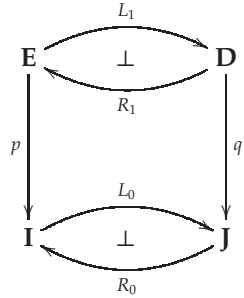
\xymatrix {
\c{E}\ar@/^1pc/[rr]^{L_1}\ar[dd]_p & \bot & \c{D}\ar@/^1pc/[ll]^{R_1}\ar[dd]^q \\
\\
\c{I}\ar@/^1pc/[rr]^{L_0} & \bot & \c{J}\ar@/^1pc/[ll]^{R_0}
}
$$
giving rise to a strong monad on a monoidal closed fibration $p$.
\begin{proposition}\label{prop:strong1}
For every $T \refs B$ and $U \refs C$, there is an expression
$$\xi : RL[B] \to \negL[{R[C]}]{\negR[{R[C]}]{B}}$$
such that
$$\eta ; \xi \deq \lc{\plugR}$$
together with a typing derivation
$$
\infer[\xi]{RL[T] \seq{\xi} \negL[{R[U]}]{\negR[{R[U]}]{T}}}{}
$$
such that
$$
\infer[C]{T \seq{\eta ; \xi} \negL[{R[U]}]{\negR[{R[U]}]{T}}}{
  \infer[\eta]{T \seq{\eta} RL[T]}{} &
  \infer[\xi]{RL[T] \seq{\xi} \negL[{R[U]}]{\negR[{R[U]}]{T}}}{}}
\deq\quad
\infer[R^\ImpL]{T \seq{\lc{\plugR}} \negL[{R[U]}]{\negR[{R[U]}]{T}}}{
 \infer[L^\ImpR]{\negR[{R[U]}]{T} \mul T \seq{\plugR} R[U]}{}
 }
$$
\end{proposition}
\begin{proof}
The expression $\xi$ is defined as the currification of 
$$
\xymatrix {
\negR[{R[C]}]{B} \mul RL[B] \ar[r]^\sigma & RL[\negR[{R[C]}]{B} \mul B] \ar[r]^{RL[\plugR]} & RLR[B] \ar[r]^{R[\epsilon]} & R[B]
}
$$
The corresponding typing derivation mirrors the structure of the expression exactly, and the equations follow from the laws of strong monads.
\end{proof}
\begin{proposition}\label{prop:tonn} For every e-type $T$ and d-type $U$, we have a derivation
$$
\infer[F^{\mu}]{\pull{\eta}RL[T] \nseq \pull{\lc{\plugR}} \negL[{R[U]}]{\negR[{R[U]}]{T}}}{}
$$
\end{proposition}
\begin{proof}
$$
\infer[R\pull{\lc{\plugR}}]{\pull{\eta}RL[T] \nseq \pull{\lc{\plugR}} \negL[{R[U]}]{\negR[{R[U]}]{T}}}{
\infer[\deq]{\pull{\eta}RL[T] \seq{\lc{\plugR}} \negL[{R[U]}]{\negR[{R[U]}]{T}}}{
\infer[C]{\pull{\eta}RL[T] \seq{\eta;\xi} \negL[{R[U]}]{\negR[{R[U]}]{T}}}{
  \infer[L\pull{\eta}]{\pull{\eta}RL[T] \seq{\eta} RL[T]}{} &
  \infer[\xi]{RL[T] \seq{\xi} \negL[{R[U]}]{\negR[{R[U]}]{T}}}{}
}}}
$$
\end{proof}
Now, to exhibit a map in the reverse direction, the discussion in the introduction to this section suggests we should ask for $RL[T]$ to be a pullback of $R[U]$.
\begin{proposition}
\label{prop:fromnn}
For every e-type $T$ and d-type $U$ such that $RL[T] \pullseq{f} R[U]$, we have a derivation
$$\infer[F^{[]}]{\pull{\lc{\plugR}} \negL[{R[U]}]{\negR[{R[U]}]{T}} \nseq \pull{\eta}RL[T]}{}$$
\end{proposition}
\begin{proof}
$$
\infer[R\pull{\eta}]{\pull{\lc{\plugR}} \negL[{R[U]}]{\negR[{R[U]}]{T}} \nseq \pull{\eta}RL[T]}{
\infer[R\pull{f}]{\pull{\lc{\plugR}} \negL[{R[U]}]{\negR[{R[U]}]{T}} \seq{\eta} RL[T]}{
\infer[\deq]{\pull{\lc{\plugR}} \negL[{R[U]}]{\negR[{R[U]}]{T}} \seq{\eta;f} R[U]}{
 \infer[C]{\pull{\lc{\plugR}} \negL[{R[U]}]{\negR[{R[U]}]{T}} \seq{\rc{\eta;f}\mul \lc{\plugR};\plugL} R[U]}{
 \infer[M]{\pull{\lc{\plugR}} \negL[{R[U]}]{\negR[{R[U]}]{T}} \seq{\rc{\eta;f} \mul \lc{\plugR}} \negR[{R[U]}]{T} \mul \negL[{R[U]}]{\negR[{R[U]}]{T}}}
  {
    \infer[R^\ImpR]{1 \seq{\rc{\eta;f}} \negR[{R[U]}]{T}}{\infer[C]{T \seq{\eta;f} {R[U]}}{\infer[\eta]{T \seq{\eta} RL[T]}{} & \infer[L\pull{f}]{RL[T] \seq{f} {R[U]}}{}}}   
  &
 \infer[L\pull{\lc{\plugR}}]{\pull{\lc{\plugR}} \negL[{R[U]}]{\negR[{R[U]}]{T}} \seq{\lc{\plugR}} \negL[{R[U]}]{\negR[{R[U]}]{T}}}{} 
  } &
  \infer[L^\ImpL]{ \negR[{R[U]}]{T} \mul \negL[{R[U]}]{\negR[{R[U]}]{T}} \seq{\plugL} {R[U]}}{}
}}}}
$$
\end{proof}
\begin{restatable}{proposition}{retractlemma}\label{prop:retract}
$F^{[]}$ is a retraction of $F^\mu$, i.e., we have (under assumption of $RL[T] \pullseq{f} R[U]$)
$$\infer[C]{\pull{\eta}RL[T] \nseq \pull{\eta}RL[T]}{
 \infer[F^{\mu}]{\pull{\eta}RL[T] \nseq \pull{\lc{\plugR}} \negL[{R[U]}]{\negR[{R[U]}]{T}}}{}
 &
 \infer[F^{[]}]{\pull{\lc{\plugR}} \negL[{R[U]}]{\negR[{R[U]}]{T}} \nseq \pull{\eta}RL[T]}{}
}\quad\deq\quad
\infer[I]{\pull{\eta}RL[T] \nseq \pull{\eta}RL[T] }{}
$$
\end{restatable}
\begin{proof}
By a long but mechanical computation.
\end{proof}
However, in general there is no reason that $F^{[]}$ has to be a section of $F^\mu$.  For example, when $R$ and $L$ are the identity and $p$ is a trivial fibration $\c{E} \to 1$, this amounts to asking that the reset map
$$
\negL[U]{\negR[U]{U}} \nseq U
$$
is an inverse to (and not just a retraction of) the shift map
$$U \nseq \negL[U]{\negR[U]{U}}$$
Considering $U = 2$ in the monoidal closed fibration $\Set \to 1$ provides an easy counterexample.

To get an isomorphism
$$\pull{\eta}RL[T] \deq \pull{\lc{\plugR}} \negL[{R[U]}]{\negR[{R[U]}]{T}}$$
we therefore need a stronger assumption.
\begin{definition}Let $S \refs A$ and $U \refs C$ be e-types.  An {\bf encoding} of $S$ in $U$ is a map $e_S : A \to C$ such that $S \pullseq{e_S} U$.  A {\bf universal type} for a type refinement system is an e-type $U$, together with an encoding $S \pullseq{e_S} U$ of $S$ in $U$ for every e-type $S$.
\end{definition}
\begin{definition}
Let $(L \dashv R) : p \dashv q$ be an adjunction of type refinement systems giving rise to a strong monad, and suppose that $U$ is a universal type in $q$, with encoding family $(e_S)_S$.  We say that $U$ is {\bf reflected} across the adjunction if:
\begin{enumerate}
\item $R$ preserves $q$-pullbacks, and
\item for each e-type $T$, the double-negation
$\negL[{R[U]}]{\negR[{R[U]}]{T}}$
is the pullback of $R[U]$ along $\rc{\eta;R[e_{L[T]}]}\mul\id;\plugR$.
\end{enumerate}
\end{definition}

\begin{theorem}
Let $(L \dashv R) : p \dashv q$ be an adjunction of type refinement systems giving rise to a strong monad, and suppose that $U$ is a universal type in $q$ reflected across the adjunction.  Then if $p$ is a monoidal closed fibration, 
$$\pull{\eta}RL[T] \deq \pull{\lc{\plugR}} \negL[{R[U]}]{\negR[{R[U]}]{T}}$$
\end{theorem}
\begin{proof}
The judgment
$$
\pull{\eta}RL[T] \seq{\eta;R[e_{L[T]}]} R[U]
$$
may be factored as the composition of two judgments
$$
\pull{\eta}RL[T] \seq{\lc{\plugR}} \negL[{R[U]}]{\negR[{R[U]}]{T}} \seq{\rc{\eta;R[e_{L[T]}]}\mul\id;\plugR} R[U]
$$
But since
$$
\pull{\eta}RL[T] \pullseq{\eta} RL[T] \pullseq{R[e_{L[T]}]} R[U]
$$
is a pullback and
$$\negL[{R[U]}]{\negR[{R[U]}]{T}} \pullseq{\rc{\eta;R[e_{L[T]}]}\mul\id;\plugR} R[U]$$
is a pullback, the left hand side
$$
\pull{\eta}RL[T] \seq{\lc{\plugR}} \negL[{R[U]}]{\negR[{R[U]}]{T}}
$$
must also be a pullback (\Cref{prop:2outof3}).
\end{proof}




\bibliographystyle{abbrvnat}

\end{document}